\documentclass{llncs}

\usepackage{graphicx}
\usepackage{amsmath}
\usepackage{amssymb}

\begin{document}
\title{ Computing the Skewness of the Phylogenetic Mean Pairwise Distance in Linear Time  }
\titlerunning{Skewness of the Phylogenetic Mean Pairwise Distance}
\author{Constantinos Tsirogiannis \and Brody Sandel}
\institute{MADALGO\thanks{Center for Massive Data Algorithmics, a Center of the Danish National Research Foundation.}
           and Department of Bioscience, Aarhus University, Denmark \\
           \email{\{constant, brody.sandel\}@cs.au.dk} }

%
%

\maketitle

\begin{abstract}  

The phylogenetic Mean Pairwise Distance ($\ensuremath{\mathrm{MPD}}$) is one of the most popular 
measures for computing the phylogenetic distance between a given group of species. 
More specifically, for a phylogenetic tree $\mathcal{T}$
and for a set of species $R$ represented by a subset of the leaf nodes of $\mathcal{T}$,
the $\ensuremath{\mathrm{MPD}}$ of $R$ is equal to the average cost of all possible simple paths in $\mathcal{T}$
that connect pairs of nodes in $R$. 

Among other phylogenetic measures, the $\ensuremath{\mathrm{MPD}}$ is used as a tool for deciding
if the species of a given group $R$ are closely related. To do this, it is important
to compute not only the value of the $\ensuremath{\mathrm{MPD}}$ for this group but also the 
expectation, the variance, and the skewness of this metric. Although 
efficient algorithms have been developed for computing the expectation 
and the variance the $\ensuremath{\mathrm{MPD}}$, there has been no approach so far for computing 
the skewness of this measure.

In the present work we describe how to compute the skewness of the $\ensuremath{\mathrm{MPD}}$ 
on a tree $\mathcal{T}$ optimally, in $\Theta(n)$ time; here $n$ is the size of the tree $\mathcal{T}$. 
So far this is the first result that leads to an exact, let alone efficient, 
computation of the skewness for any popular phylogenetic distance measure. 
Moreover, we show how we can compute in $\Theta(n)$ time several interesting 
quantities in $\mathcal{T}$ that can be possibly used as building blocks for computing 
efficiently the skewness of other phylogenetic measures.
  
\end{abstract}

\thispagestyle{empty}

\section{Introduction}

Communities of co-occuring species may be 
described as ``clustered'' if species in the community tend to be 
close phylogenetic relatives of one another, 
or ``overdispersed'' if they are distant relatives~\cite{webb1}. 
To define these terms we need a function that measures the phylogenetic 
relatedness of a set of species, and also a point of reference for how this 
function should behave in the absence of ecological and evolutionary 
processes. One such function is the mean pairwise distance ($\ensuremath{\mathrm{MPD}}$);
given a phylogenetic tree $\mathcal{T}$ and a subset of species $R$ that are represented
by leaf nodes of $\mathcal{T}$, the $\ensuremath{\mathrm{MPD}}$ of the species in $R$ is equal to
average cost of all possible simple paths that connect pairs of nodes in $R$.  

To decide if the value of the $\ensuremath{\mathrm{MPD}}$ for a specific set of species $R$ is large or small,
we need to know the average value (expectation) of the $\ensuremath{\mathrm{MPD}}$ for all sets of species 
in $\mathcal{T}$ that consist of exactly $r=|R|$ species. To judge how much larger or 
smaller is this value from the average, we also need to know the standard 
deviation of the $\ensuremath{\mathrm{MPD}}$ for all possible sets of $r$ species in $\mathcal{T}$.
Putting all these values together, we get the following index that expresses 
how clustered are the species in $R$~\cite{webb1}:
\[
\ensuremath{\mathrm{NRI}} = \frac{ \ensuremath{\mathrm{MPD}}(\mathcal{T},R) - \ensuremath{\mathrm{expec}}_{\ensuremath{\mathrm{MPD}}}(\mathcal{T},r) }{sd_\ensuremath{\mathrm{MPD}}(\mathcal{T},r)},
\]

where  $\ensuremath{\mathrm{MPD}}(\mathcal{T},R)$ is the value of the $\ensuremath{\mathrm{MPD}}$ for $R$ in $\mathcal{T}$, and $\ensuremath{\mathrm{expec}}(\mathcal{T})$ 
and $sd_\ensuremath{\mathrm{MPD}}(\mathcal{T},r)$ are the expected value and the standard deviation respectively of the 
$\ensuremath{\mathrm{MPD}}$ calculated over all subsets of $r$ species in $\mathcal{T}$. 

 In a previous paper we presented optimal algorithms for computing the expectation and the 
standard deviation of the $\ensuremath{\mathrm{MPD}}$ of a phylogenetic tree $\mathcal{T}$ in $\Theta(n)$ time, where~$n$ 
is the number of the edges of $\mathcal{T}$~\cite{ourpaper}. This enabled exact computations of these statistical
moments of the $\ensuremath{\mathrm{MPD}}$ on large trees, which were previously infeasible using traditional 
slow and inexact resampling techniques. However, one important problem remained unsolved; 
quantifying our degree of confidence that the $\ensuremath{\mathrm{NRI}}$ value observed in a community reflects 
non-random ecological and evolutionary processes. 

This degree of confidence is a statistical $P$ value, that is 
the probability that we would observe an $\ensuremath{\mathrm{NRI}}$ value as extreme or more so if the community were randomly
assembled. Traditionally, estimating $P$ is accomplished by ranking the observed
$\ensuremath{\mathrm{MPD}}$ against the distribution of randomized $\ensuremath{\mathrm{MPD}}$ values~\cite{pontarp}. If the
$\ensuremath{\mathrm{MPD}}$ falls far enough into one of the tails of the distribution (generally below the 2.5
percentile or above the 97.5 percentile, yielding $P < 0.05$), the community is said to be
significantly overdispersed or significantly clustered. However, this approach relies
on sampling a large number of random subsets of species in $\mathcal{T}$, and recomputing the 
$\ensuremath{\mathrm{MPD}}$ for each random subset. Therefore, this method is slow and imprecise. 

 We can approximate the $P$ value of an observed $\ensuremath{\mathrm{NRI}}$ by assuming a particular distribution
of the possible $\ensuremath{\mathrm{MPD}}$ values and evaluating its cumulative distribution function at the
observed $\ensuremath{\mathrm{MPD}}$. Because the $\ensuremath{\mathrm{NRI}}$ measures the difference between the observed values and
expectation in units of standard deviations, this yields a very simple rule if we assume
that possible $\ensuremath{\mathrm{MPD}}$ values are normally distributed: any $\ensuremath{\mathrm{NRI}}$ value larger than $1.96$ or 
smaller than $-1.96$ is significant. Unfortunately, the distribution of $\ensuremath{\mathrm{MPD}}$ values is 
often skewed, such that this simple rule will lead to incorrect $P$ value estimates~\cite{cooper,vamosi}. 
Of particular concern, this skewness introduces a bias towards 
detecting either significant clustering or significant overdispersion~\cite{harmon}. 
Calculating this skewness analytically would enable us to remove this bias and 
improve the accuracy of $P$ value estimates obtained analytically. However, so far there has been
no result in the related literature that shows how to compute this skewness value.

Hence, given a phylogenetic tree $\mathcal{T}$ and an integer $r$ there is the need to design 
an efficient and exact algorithm that can compute the skewness of the $\ensuremath{\mathrm{MPD}}$ 
for $r$ species in $\mathcal{T}$. This would provide the last critical piece required for the adoption 
of a fully analytical and efficient approach for analysing ecological communities using 
the $\ensuremath{\mathrm{MPD}}$ and the $\ensuremath{\mathrm{NRI}}$.

\paragraph{Our Results} In the present work we tackle the problem of computing efficiently 
the skewness of the $\ensuremath{\mathrm{MPD}}$. More specifically, given a tree $\mathcal{T}$ that consists of $n$ edges 
and a positive integer $r$, we prove that we can compute the skewness of
of the $\ensuremath{\mathrm{MPD}}$ over all subsets of $r$ leaf nodes in $\mathcal{T}$ optimally, in $\Theta(n)$ time.
For the calculation of this skewness value we consider that every subset of exactly $r$ species 
in $\mathcal{T}$ is picked uniformly out of all possible subsets that have $r$ species.
The main contribution of this paper is a constructive proof that leads straightforwardly 
to an algorithm that computes the skewness of the $\ensuremath{\mathrm{MPD}}$ in $\Theta(n)$ time. 
This is clearly very efficient, especially if we consider that it outperforms the best 
algorithms that are known so far for computing lower-order statistics for other phylogenetic measures;
for example the most efficient known algorithm for computing the variance of 
the popular Phylogenetic Distance ($\ensuremath{\mathrm{PD}}$) runs in $O(n^2)$ time~\cite{ourpaper}. 

More than that, we prove how we can compute in $\Theta(n)$ time several quantities that are related
with groups of paths in the given tree; these quantities can be possibly used 
as building blocks for computing efficiently the skewness (and other statistical moments) 
of phylogenetic measures that are similar to the $\ensuremath{\mathrm{MPD}}$. Such an example is the measure 
which is the equivalent of the $\ensuremath{\mathrm{MPD}}$ for computing the distance between two subsets of 
species in $\mathcal{T}$~\cite{swenson}.

The rest of this paper is, in its entirety, an elaborate proof for computing the 
skewness of the $\ensuremath{\mathrm{MPD}}$ on a tree $\mathcal{T}$ in $\Theta(n)$ time. In the next section 
we define the problem that we want to tackle, and we present a group of quantities
that we use as building blocks for computing the skewness of the $\ensuremath{\mathrm{MPD}}$. We prove 
that all of these quantities can be computed in linear time with respect to the size
of the input tree. In Section~\ref{sec::mpd_proof} we provide the main proof of this 
paper; there we show how we can express the value of the skewness of the $\ensuremath{\mathrm{MPD}}$ in 
terms of the quantities that we introduced earlier. The proof implies a straightforward 
linear time algorithm for the computation of the skewness as well.

\section{ Description of the Problem and Basic Concepts }\label{sec::basic}

\paragraph{Definitions and Notation} Let $\mathcal{T}$ be a phylogenetic tree, and let $E$ be the set of its edges.
We denote the number of the edges in $\mathcal{T}$ by $n$, that is $n = |E|$. 
For an edge $e \in E$, we use $w_e$ to indicate the weight of this edge.
We use $S$ to denote the set of the leaf nodes of $\mathcal{T}$. 
We call these nodes the \emph{tips} of the tree, and we use $s$ to denote
the number of these nodes.    

Since a phylogenetic tree is a rooted tree, for any edge $e \in E$ we distinguish the two 
nodes adjacent to $e$ into a \emph{parent} node and a \emph{child} node; among these two, 
the parent node of $e$ is the one for which the simple path from this node to the root 
does not contain $e$. We use $\ensuremath{\mathrm{Ch}}(e)$ to indicate the set of edges whose parent node is
the child node of $e$, which of course implies that $e \notin \ensuremath{\mathrm{Ch}}(e)$. We indicate the
edge whose child node is the parent node of $e$ by $\ensuremath{\mathrm{parent}}(e)$. 
For any edge $e \in E$, tree $\mathcal{T}(e)$ is the subtree of $\mathcal{T}$ whose 
root is the child node of edge $e$. We denote the set of tips that appear in 
$\mathcal{T}(e)$ as $S(e)$, and we denote the number of these tips by $s(e)$.

Given any edge $e \in E$, we partition the edges of $\mathcal{T}$ into three subsets. The first subset
consists of all the edges that appear in the subtree of $e$. We denote this set by $\ensuremath{\mathrm{Off}}(e)$. 
The second subset consists of all edges $e' \in E$ for which $e$ appears in the subtree of $e'$. 
We use $\ensuremath{\mathrm{Anc}}(e)$ to indicate this subset. For the rest of this paper, we define that $e \in \ensuremath{\mathrm{Anc}}(e)$, 
and that $e \notin \ensuremath{\mathrm{Off}}(e)$. The third subset contains all the tree edges that do not appear 
neither in $\ensuremath{\mathrm{Off}}(e)$, nor in $\ensuremath{\mathrm{Anc}}(e)$; we indicate this subset by $\ensuremath{\mathrm{Ind}}(e)$.

For any two tips $u,v \in S$, we use $p(u,v)$ to indicate the simple path in $\mathcal{T}$ between these nodes.
Of course, the path $p(u,v)$ is unique since $\mathcal{T}$ is a tree. We use $cost(u,v)$ to denote the cost 
of this path, that is the sum of the weights of all the edges that appear on the path. 
Let $u$ be a tip in $S$ and let $e$ be an edge in $E$. We use $cost(u,e)$ to represent
the cost of the shortest simple path between $u$ and the child node of $e$. Therefore,
if $u \in S(e)$ this path does not include $e$, otherwise it does.  
For any subset $R \subseteq S$ of the tips of the tree $\mathcal{T}$,  
we denote the set of all pairs of elements in $R$, that is the set of all 
combinations that consist of two distinct tips in $R$, by $\Delta(R)$. 
Given a phylogenetic tree $\mathcal{T}$ and a subset of its tips $R \subseteq S$,
we denote the Mean Pairwise Distance of $R$ in $\mathcal{T}$ by $\ensuremath{\mathrm{MPD}}(\mathcal{T}, R)$. 
Let $r = |R|$. This measure is equal to:
%
\[\ensuremath{\mathrm{MPD}}(\mathcal{T}, R) = \frac{2}{r(r-1)}\sum_{ \{u,v\} \in \Delta(R)}cost(u,v) \ .  \]
%

\subsection{ Aggregating the Costs of Paths } 

Let $\mathcal{T}$ be a phylogenetic tree that consists of $n$ edges and $s$ tips, and let
$r$ be a positive integer such that $r \leq s$. We use $\ensuremath{\mathrm{sk}}(\mathcal{T},r)$ to denote the skewness
of the $\ensuremath{\mathrm{MPD}}$ on $\mathcal{T}$ when we pick a subset of $r$ tips of this tree with
uniform probability. In the rest of this paper we describe 
in detail how we can compute $\ensuremath{\mathrm{sk}}(\mathcal{T},r)$ in $O(n)$ time, by scanning $\mathcal{T}$ only a 
constant number of times. Based on the formal definition
of the concept of skewness, the value of $\ensuremath{\mathrm{sk}}(\mathcal{T},r)$ is equal to:
%

\begin{align}
 \ensuremath{\mathrm{sk}}(\mathcal{T},r) & =  E_{R \in \mathrm{Sub}(S, r)}\left[ \left(\frac{\ensuremath{\mathrm{MPD}}(\mathcal{T},R)-\ensuremath{\mathrm{expec}}(\mathcal{T},r)}{\ensuremath{\mathrm{var}}(\mathcal{T},r)}\right)^3 \right] \hfill\notag\\[.2em]
& = \frac{E_{R \in \mathrm{Sub}(S, r)}[\ensuremath{\mathrm{MPD}}^3(\mathcal{T},R)]-3 \cdot \ensuremath{\mathrm{var}}(\mathcal{T},r)^2-\ensuremath{\mathrm{expec}}(\mathcal{T},r)^3}{\ensuremath{\mathrm{var}}(\mathcal{T},r)^3} \ , \label{eq::skewness_basic} 
\end{align}
%
where $\ensuremath{\mathrm{expec}}(\mathcal{T},r)$ and $\ensuremath{\mathrm{var}}(\mathcal{T},r)$ are the expectation and the variance of the $\ensuremath{\mathrm{MPD}}$ for
subsets of exactly $r$ tips in $\mathcal{T}$, and $E_{R \in \mathrm{Sub}(S, r)}[\cdot]$ denotes the function of the expectation
over all subsets of exactly $r$ tips in $S$. 
In a previous paper, we showed how we can compute the expectation and the variance of the $\ensuremath{\mathrm{MPD}}$ on $\mathcal{T}$ in $O(n)$ 
time~\cite{ourpaper}. Therefore, in the rest of this work we focus on analysing the value
$E_{R \in \mathrm{Sub}(S, r)}[\ensuremath{\mathrm{MPD}}^3(\mathcal{T},R)]$ and expressing this quantity in a way that can be computed efficiently, 
in linear time with respect to the size of $\mathcal{T}$. 

To make things more simple, we break the 
description of our approach into two parts; in the first part, we define several 
quantities that come from adding and multiplying the costs of specific subsets of paths between 
tips of the tree. We also present how we can compute all these quantities in $O(n)$ time
in total by scanning $\mathcal{T}$ a constant number of times. Then, in Section~\ref{sec::mpd_proof}, 
we show how we can express the skewness of the $\ensuremath{\mathrm{MPD}}$ on $\mathcal{T}$ based on 
these quantities, and hence compute the skewness in $O(n)$ time as well. 
Next we provide the quantities that we want to consider in our analysis; these quantities
are described in Table~\ref{tab::quantities}. 

\begin{table}[h!]
\caption{The quantities that we use for expressing the skewness of the $\ensuremath{\mathrm{MPD}}$.} \label{tab::quantities}
\resizebox{\linewidth}{!}{
\begin{tabular}{|l||l|}
\hline
I) \( \displaystyle \ensuremath{\mathrm{TC}}(\mathcal{T}) = \sum_{\{u,v\} \in \Delta(S)} cost(u,v) \) &
II) \( \displaystyle \ensuremath{\mathrm{CB}}(\mathcal{T}) = \sum_{\{u,v\} \in \Delta(S)} cost^3(u,v) \)
\\
\hline
III) \( \displaystyle \forall e \in E, \ \ensuremath{\mathrm{TC}}(e) = \sum_{ \substack{ \{u,v\} \in \Delta(S) \\ e \in p(u,v) }} 
        cost(u,v)  \)
  & 
IV) \( \displaystyle \forall e \in E, \ \ensuremath{\mathrm{SQ}}(e) = 
      \sum_{ \substack{ \{u,v\} \in \Delta(S) \\ e \in p(u,v) }} cost^2(u,v) \)
\\
\hline
V) \( \displaystyle \forall e \in E, \ \ensuremath{\mathrm{Mult}}(e) = 
      \sum_{ \substack{ \{u,v\} \in \Delta(S) \\ e \in p(u,v) }} \ensuremath{\mathrm{TC}}(u) \cdot \ensuremath{\mathrm{TC}}(v) \)
  &
VI) \( \displaystyle \forall u \in S, \ \ensuremath{\mathrm{SM}}(u) = \sum_{v\in S \setminus \{u\}}cost(u,v) \cdot \ensuremath{\mathrm{TC}}(v)  \)
\\
\hline
VII) \( \displaystyle \forall e \in E, \ \ensuremath{\mathrm{TC_{sub}}}(e) = \sum_{ \substack{ u \in S(e) }} cost(u,e) \)
 &
VIII) \( \displaystyle \forall e \in E, \ \ensuremath{\mathrm{SQ_{sub}}}(e) = \sum_{ \substack{ u \in S(e) }} cost^2(u,e) \)
\\
\hline
IX) \( \displaystyle \forall e \in E, \ \ensuremath{\mathrm{PC}}(e) = \sum_{ \substack{ u \in S }} cost(u,e) \)
 &    
X) \( \displaystyle \forall e \in E, \ \ensuremath{\mathrm{PSQ}}(e) = \sum_{ \substack{ u \in S }} cost^2(u,e) \)
\\\hline
XI) \( \displaystyle \forall e \in E, \ \ensuremath{\mathrm{QD}}(e) = \sum_{ u \in S(e) } 
       \left( \sum_{v \in S(e) \setminus \{ u \} } cost(u,v) \right)^{2} \)
 &\phantom{$\Bigg(^{\Big)}_\Big)$} \\\hline
\end{tabular}}
\end{table}

For any tip $u \in S$, we define that $\ensuremath{\mathrm{SQ}}(u) = \ensuremath{\mathrm{SQ}}(e)$, and $\ensuremath{\mathrm{TC}}(u) = \ensuremath{\mathrm{TC}}(e)$,
where $e$ is the edge whose child node is $u$. 
The proof of the following lemma is provided in the full version of this paper.

\begin{lemma}~\label{le::fast_quantities}
Given a phylogenetic tree $\mathcal{T}$ that consists of $n$ edges, we can compute all the 
quantities that are presented in Table~\ref{tab::quantities} in $O(n)$ time in total.
\end{lemma}
\section{Computing the Skewness of the MPD}\label{sec::mpd_proof}

In the previous section we defined the problem of computing the skewness of the $\ensuremath{\mathrm{MPD}}$
for a given phylogenetic tree $\mathcal{T}$. Given a positive integer $r \leq s$, 
we showed that to solve this problem efficiently it remains to find an efficient algorithm 
for computing $E_{R \in \mathrm{Sub}(S, r)}[\ensuremath{\mathrm{MPD}}^3(\mathcal{T},R)]$; this is the mean value of the cube
of the $\ensuremath{\mathrm{MPD}}$ among all possible subsets of tips in $\mathcal{T}$ that consist of exactly $r$ elements.
To compute this efficiently, we introduced in Table~\ref{tab::quantities} ten different 
quantities which we want to use in order to express this mean value. In 
Lemma~\ref{le::fast_quantities} we proved that these quantities can be computed in $O(n)$ time,
where $n$ is the size of $\mathcal{T}$. 

Next we prove how we can calculate the value for the 
mean of the cube of the $\ensuremath{\mathrm{MPD}}$ based on the quantities in Table~\ref{tab::quantities}. 
In particular, in the proof of the following lemma we show how the value $E_{R \in \mathrm{Sub}(S, r)}[\ensuremath{\mathrm{MPD}}^3(\mathcal{T},R)]$ 
can be written analytically as an expression that contains the quantities in 
Table~\ref{tab::quantities}. This expression can then be straightforwardly 
evaluated in $O(n)$ time, given that we have already computed the aforementioned 
quantities~\footnote{Because the full form of this expression is very long (it consists of 
a large number of terms), we have chosen not to include it in the definition of the following lemma. 
We chose to do so because we considered that including the entire expression would not 
make this work more readable. In any case, the full expression can be easily infered from the proof of 
the lemma}.

\begin{lemma}\label{le::main_lemma}
For any given natural $r \leq s$, we can compute $E_{R \in \mathrm{Sub}(S, r)}[\ensuremath{\mathrm{MPD}}^3(\mathcal{T},R)]$ in $\Theta(n)$ time.
\end{lemma}
\begin{proof}
The expectation of the cube of the $\ensuremath{\mathrm{MPD}}$ is equal to:\\
%
\begin{align*}
&E_{R \in \mathrm{Sub}(S, r)}[\ensuremath{\mathrm{MPD}}^3(\mathcal{T},R)] = \frac{8}{r^3(r-1)^3} \ \ \cdot \\
&E_{R \in \mathrm{Sub}(S, r)} \left[
\sum_{{\{u,v\}\in\Delta(R)}} \hspace{0.05in} \sum_{{\{x,y\}\in\Delta(R)}}
\hspace{0.05in}\sum_{{\{c,d\}\in\Delta(R)}}
cost(u,v)\cdot cost(x,y)\cdot cost(c,d)\right].
\end{align*}
%
From the last expression we get:
\begin{align}
& E_{R \in \mathrm{Sub}(S, r)} \Bigg[
\sum_{\{u,v\}\in\Delta(R)} \hspace{0.05in} \sum_{ \{x,y\} \in \Delta(R) }
\hspace{0.05in} \sum_{ \{c,d\} \in \Delta(R) }
cost(u,v) \cdot cost(x,y) \cdot cost(c,d) \Bigg]\nonumber \\[.2em]
&=\hspace{-.7em}\sum_{{\{u,v\}\in\Delta(S)}} \sum_{{\{x,y\}\in\Delta(S)}} \sum_{{\{c,d\}\in\Delta(S)}}
\!\!cost(u,\!v)\!\cdot \!cost(x,\!y) \!\cdot \!cost(c,d) \ \ \cdot \nonumber\\[.2em]
& \ \ E_{R \in \mathrm{Sub}(S, r)}[AP_{\!R}(u,\!v,x,y,c,d)]\,,\label{eq::first_big_sum}
\end{align}
where $AP_R(u,v,x,y,c,d)$ is a random variable whose value is equal to one in the case that 
$u,v,x,y,c,d \in R$, otherwise it is equal to zero. For any six tips $u,v,x,y,c,d \in S$, which 
may not be all of them distinct, we use
$\theta(u,v,x,y,c,d)$ to denote the number of distinct elements among these tips.
Let $t$ be an integer, and let $(t)_k$ denote the $k$-th falling factorial power of $t$, 
which means that $(t)_k = t(t-1)\ldots(t-k+1)$.
For the expectation of the random variables that appear in the last expression it holds that:
%
\begin{align}
& E_{R \in \mathrm{Sub}(S, r)}\left[ AP_R(u,v,x,y,c,d)  \right]  
= \frac{(r)_{\theta(u,v,x,y,c,d)}}{(s)_{\theta(u,v,x,y,c,d)}} \label{eq::apr_theta}
\end{align}
%
Notice that in~(\ref{eq::apr_theta}) we have $ 2 \leq \theta(u,v,x,y,c,d) \leq 6$. The value
of the function $\theta(\cdot)$ cannot be smaller than two in the above case because we have that
$u \neq v $, $x \neq y $, and $c \neq d$. Thus, we can rewrite (\ref{eq::first_big_sum}) as:
%
\begin{align}
\sum_{{\{u,v\}\in\Delta(S)}} \ \sum_{{\{x,y\}\in\Delta(S)}}\ \sum_{{\{c,d\}\in\Delta(S)}} \hspace{0.03in} \frac{(r)_{\theta(u,v,x,y,c,d)}}{(s)_{\theta(u,v,x,y,c,d)}}
\cdot cost(u,\!v) \cdot cost(x,y) \cdot cost(c,d)
\label{eq::second_big_sum}
\end{align}
%
Hence, our goal now is to compute a sum whose elements are the product of costs of triples of paths.
Recall that for each of these paths, the end-nodes of the path are a pair of distinct tips in the tree.
Although the end-nodes of each path are distinct, in a given triple the paths may share one or more 
end-nodes with each other. Therefore, the distinct tips in any triple of paths may vary from 
two up to six tips. Indeed, in (\ref{eq::second_big_sum}) we get a sum where the triples of paths in 
the sum are partitioned in five groups; a triple of paths is assigned to a group depending on the number
of distinct tips in this triple. In~(\ref{eq::second_big_sum}) the sum for each group of triples is 
multiplied by the same factor $(r)_{\theta(u,v,x,y,c,d)}/(s)_{\theta(u,v,x,y,c,d)}$, 
hence we have to calculate the sum for each group of triples separately. 

However, when we try to calculate the sum for each of these groups 
of triples we see that this calculation is more involved; some of these groups of triples are 
divided into smaller subgroups, depending on which end-nodes of the paths in each triple are the same.
To explain this better, we can represent a triple of paths schematically as a graph; 
let $\{u,v\}, \{x,y\}, \{c,d\} \in \Delta(S)$ be three pairs of tips in $\mathcal{T}$. As mentioned already, the tips
within each pair are distinct, but tips between different pairs can be the same.  
\begin{figure}[h!]
\centering
\includegraphics[width=2.0in]{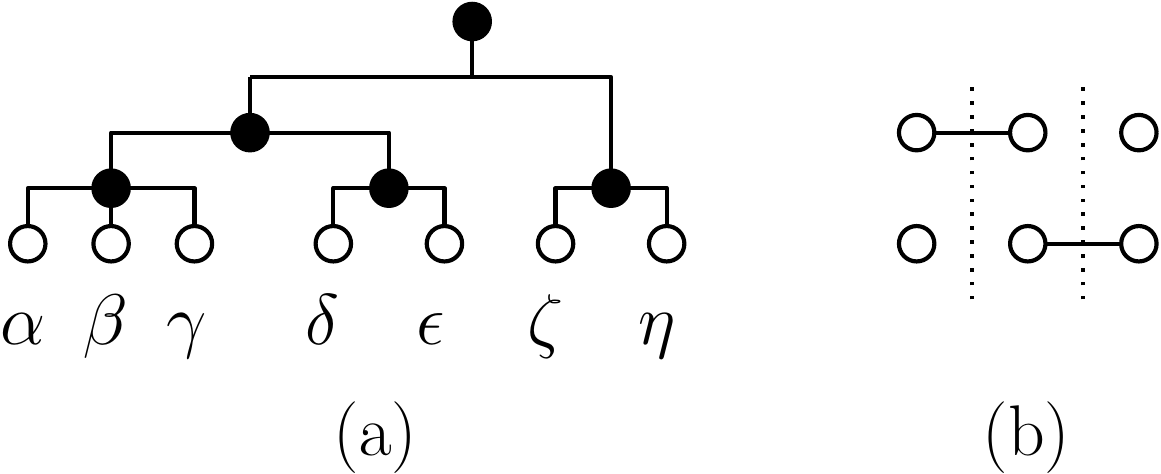}
\caption{ (a) A phylogenetic tree $\mathcal{T}$  and (b) an example of the tripartite graph induced by the triplet of 
its tip pairs $\{\alpha, \gamma\},$  $\{\delta, \gamma\},$ $\{\epsilon, \delta \},$ , where 
$\{\alpha, \gamma, \delta, \epsilon \} \subset S$. The dashed lines in the graph distinguish 
the partite subsets of vertices; the vertices of each
partite subset correspond to tips of the same pair.}
 \label{fig::example_graph}
\end{figure}
We represent the similarity between tips of these three pairs as a graph of six vertices.
Each vertex in the graph corresponds to a tip of these three pairs. Also, there exists an edge in this graph
between two vertices if the corresponding tips are the same. Thus, this graph is tripartite;
no vertices that correspond to tips of the same pair can be connected to each other with an edge. 
Hence, we have a tripartite graph where each partite set of vertices consists of two 
vertices--see Fig.~\ref{fig::example_graph} for an example. 

For any triple of pairs of tips $\{u,v\}$, $\{x,y\}$, $\{c,d\} \in \Delta(S)$ we denote
the tripartite graph that corresponds to this triple by $G[u,v,x,y,c,d]$. We call this graph
the \emph{similarity} graph of this triple. Based on the way that similarities may occur 
between tips in a triple of paths, we can 
partition the five groups of triples in~(\ref{eq::second_big_sum}) into smaller subgroups.
Each of these subgroups contains triples whose similarity graphs are isomorphic.  
For a tripartite graph that consists of three partite sets of two vertices each, 
there can be eight different isomorphism classes. Therefore, the five groups of triples are 
partitioned into eight subgroups. Figure~\ref{fig::isomorphisms} illustrates the eight isomorphism
classes that exist for the specific kind of tripartite graphs that we consider.
Since we refer to isomorphism classes, each of the graphs in Fig.~\ref{fig::isomorphisms} represents 
the combinatorial structure of the similarities between three pairs of tips, and it does not 
correspond to a particular planar embedding, or ordering of the tips.

\begin{figure}[h!]
\centering
\includegraphics[width=3.6in]{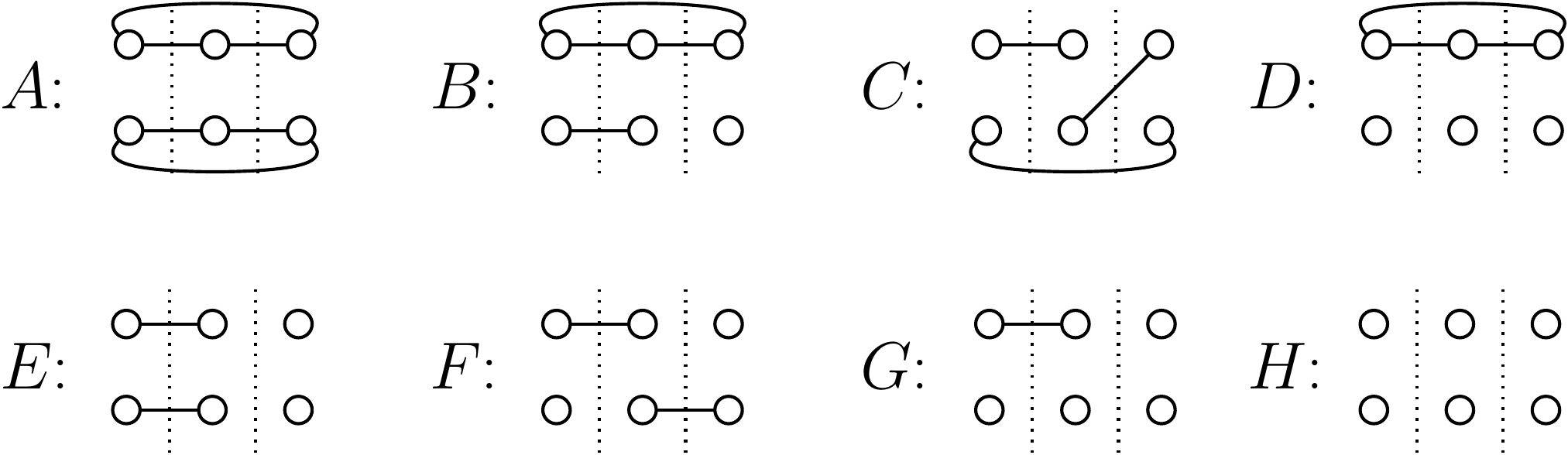}
\caption{ The eight isomorphism classes of a tripartite graph of $3 \times 2$ vertices
          that represent schematically the eight possible cases of similarities 
          between tips that we can have when we consider three paths between pairs of tips in a tree $\mathcal{T}$.}
\label{fig::isomorphisms}
\end{figure}

Let $X$ be any isomorphism class that is illustrated in Figure~\ref{fig::isomorphisms}.
We denote the set of all triples of pairs in $\Delta(S)$ whose similarity graphs belong to this class
by $\mathcal{B}_X$. More formally, the set $\mathcal{B}_X$ can be defined as follows :
%
\begin{align*}
\mathcal{B}_X = \{ \textrm{ } \{ \{u,v\}, \{x,y\}, \{c,d\} \} : \{u,v\}, \{x,y\}, \{c,d\} \in \Delta(S) \\[0.3em]
\textrm{ and } G[u,v,x,y,c,d] \textrm{ belongs to class $X$ in Figure~\ref{fig::isomorphisms} } \}  \ .
\end{align*}
%
We introduce also the following quantity:
\[
\ensuremath{\mathrm{TRS}}(X) = \sum_{ \{ \{u,v\}, \{x,y\}, \{c,d\} \} \in \mathcal{B}_X } cost(u,v) \cdot cost(x,y) \cdot cost(c,d) \ .
\]
Hence, we can rewrite~(\ref{eq::second_big_sum}) as follows:
{\small
\begin{align}
&  \frac{(r)_2}{(s)_2} \cdot \ensuremath{\mathrm{TRS}}(A) + 3 \cdot \frac{(r)_3}{(s)_3} \cdot \ensuremath{\mathrm{TRS}}(B) 
+ 6 \cdot \frac{(r)_3}{(s)_3} \cdot \ensuremath{\mathrm{TRS}}(C) 
+ 6 \cdot \frac{(r)_4}{(s)_4} \cdot \ensuremath{\mathrm{TRS}}(D) \nonumber \\[0.2em]
& +  3 \cdot \frac{(r)_4}{(s)_4}\cdot \ensuremath{\mathrm{TRS}}(E) + 6 \cdot \frac{(r)_4}{(s)_4} \cdot \ensuremath{\mathrm{TRS}}(F)
+ 6 \cdot \frac{(r)_5}{(s)_5} \cdot \ensuremath{\mathrm{TRS}}(G) + 6 \cdot \frac{(r)_6}{(s)_6} \cdot \ensuremath{\mathrm{TRS}}(H)
 \label{eq::le_big_sum}
\end{align}
}%

Notice that some of the terms $\frac{(r)_i}{(s)_i} \cdot \ensuremath{\mathrm{TRS}}(X)$ in~(\ref{eq::le_big_sum}) are 
multiplied with an extra constant factor. 
This happens for the following reason; the sum in $\ensuremath{\mathrm{TRS}}(X)$ counts each triple once for every 
different combination of three pairs of tips. However, in the triple sum in~(\ref{eq::second_big_sum})
some triples appear more than once. For example, every triple that belongs in class $B$ appears
three times in~(\ref{eq::second_big_sum}), hence there is an extra factor three in front of $\ensuremath{\mathrm{TRS}}(B)$
in~(\ref{eq::le_big_sum}).   
 
To compute efficiently $E_{R \in \mathrm{Sub}(S, r)}[\ensuremath{\mathrm{MPD}}^3(\mathcal{T},R)]$, it remains to 
compute efficiently each value $\ensuremath{\mathrm{TRS}}(X)$ for every isomorphism class $X$ that is 
presented in Figure~\ref{fig::isomorphisms}. Next we show in detail how we can do that 
by expressing each quantity $\ensuremath{\mathrm{TRS}}(X)$ as a function of the quantities that appear in Table~\ref{tab::quantities}.

For the triples that correspond to the isomorphism class $A$ we have:
%
\begin{align*}
\ensuremath{\mathrm{TRS}}(A) = \sum_{ \{u,v\} \in \Delta(S) } cost^3(u,v) = \ensuremath{\mathrm{CB}}(\mathcal{T}) \ .
\end{align*}
%
%
For $\ensuremath{\mathrm{TRS}}(B)$ we get:\vspace{.8em}\\
\resizebox{\linewidth}{!}{%
$\begin{array}{rl}
\displaystyle\ensuremath{\mathrm{TRS}}(B)&=\displaystyle\!\!\!\sum_{{\{u,v\}\in\Delta(S)}}\!\!\!cost^2(u,v)\left(
\sum_{x \in S \setminus \{u\} }\!\!cost(u,x) +\!\!\!\!\sum_{y \in S \setminus \{v\} }\!\!cost(v,y) - 2 \cdot cost(u,v) \right) \vspace{1em}\\
 &\displaystyle=\!\!\sum_{{\{u,v\}\in\Delta(S)}} cost^2(u,v) \left( \ensuremath{\mathrm{TC}}(u) + \ensuremath{\mathrm{TC}}(v) - 2 \cdot cost(u,v) \right)\vspace{1em}\\
 &\displaystyle=\!\!\sum_{u \in S }\ensuremath{\mathrm{SQ}}(u) \cdot \ensuremath{\mathrm{TC}}(u) - 2 \cdot \ensuremath{\mathrm{CB}}(\mathcal{T}) \ .
\end{array}$
}\bigskip\\
%
The quantity $\ensuremath{\mathrm{TRS}}(C)$ is equal to:
\begin{align}
& \frac{1}{6} \sum_{u \in S} \hspace{0.05in} \sum_{v \in S \setminus \{u\}}cost(u,v) 
\sum_{ x \in S \setminus \{u,v\} } cost(u,x)\cdot cost(x,v) \nonumber \\[0.2em]
 = & \frac{1}{6} \sum_{e \in E} w_e  \sum_{u \in S(e)} \hspace{0.05in} \sum_{v \in S - S(e)} \hspace{0.05in} \sum_{x \in S \setminus \{u,v\}} cost(u,x)\cdot cost(x,v) \ . \label{eq::nasty_0}   
\end{align}
For any $e \in E$ we have that:
\begin{align}
& \sum_{u \in S(e)} \hspace{0.05in} \sum_{v \in S - S(e)} \hspace{0.05in} 
\sum_{x \in S \setminus \{u,v\}} cost(u,x)\cdot cost(x,v) \nonumber\\[0.2em]
 &= \sum_{u \in S(e)} \hspace{0.05in} \sum_{v \in S \setminus \{u\}} \hspace{0.05in} 
 \sum_{x \in S \setminus \{u,v\}} cost(u,x) \cdot cost(x,v) \label{eq::nasty_1a} \\[0.2em]
&- 2 \sum_{\{u,v\}\in\Delta(S(e))}\hspace{0.05in} \sum_{x \in S \setminus \{u,v\} } cost(u,x) \cdot cost(x,v) \ .\tag{\ref{eq::nasty_1a}b} \label{eq::nasty_1b}
\end{align}
The first of the two sums in~(\ref{eq::nasty_1a}) can be written as:
\begin{align}
&\sum_{u \in S(e)} \hspace{0.05in} \sum_{v \in S \setminus \{u\}} \hspace{0.05in}
\sum_{x \in S \setminus \{u,v\}} cost(u,x) \cdot cost(x,v) \nonumber \\[0.2em]
	&=\sum_{u \in S(e)} \hspace{0.05in} \sum_{v \in S \setminus \{u\}} \hspace{0.05in}
\sum_{x \in S \setminus \{u,v\}} cost(u,v) \cdot cost(v,x) \nonumber \\[0.2em]
	&= \sum_{u \in S(e)} \hspace{0.05in} \sum_{v \in S \setminus \{u\} } (cost(u,v)
\cdot \ensuremath{\mathrm{TC}}(v) - cost^2(u,v)) \nonumber \\[0.2em]
	&= \sum_{u \in S(e)} \ensuremath{\mathrm{SM}}(u) - \ensuremath{\mathrm{SQ}}(u) \ .
\end{align}
According to Lemma~\ref{le::main_lemma}, we can compute $\ensuremath{\mathrm{SM}}(u)$ and $\ensuremath{\mathrm{SQ}}(u)$ for all tips $u \in S$ 
in linear time with respect to the size of $\mathcal{T}$. Given these values, we can compute 
$\sum_{u \in S(e)} \ensuremath{\mathrm{SM}}(u) - \ensuremath{\mathrm{SQ}}(u)$ for every edge $e \in E$ in $\mathcal{T}$ with a single bottom-up scan
of the tree.
For any edge $e$ in $E$, the second sum in~(\ref{eq::nasty_1b}) is equal to:
\begin{align}
& \sum_{ \{u, v\} \in \Delta(S(e))} \hspace{0.05in} 
\sum_{ x \in S \setminus \{u ,v\} } cost(u,x) \cdot cost(x,v) \nonumber \\[0.2em]
&= \sum_{ \{u, v\} \in \Delta(S(e))} \hspace{0.05in}
\sum_{ x \in S(e) \setminus \{u ,v\} }  cost(u,x) \cdot cost(x,v) \label{eq::nasty_2a} \\[0.2em]
&+ \sum_{ \{u, v\} \in \Delta(S(e))} \hspace{0.05in}
\sum_{ x \in S \setminus S(e) }  cost(u,x) \cdot cost(x,v) \ .
\tag{\ref{eq::nasty_2a}b}\label{eq::nasty_2b}
\end{align}
We can express the first sum in~(\ref{eq::nasty_2a}) as:
\begin{align}
&\sum_{ \{u, v\} \in \Delta(S(e))} \hspace{0.05in}
\sum_{ x \in S(e) \setminus \{u ,v\} }  cost(u,x) \cdot cost(x,v)   \nonumber \\[0.2em]
&= \frac{1}{2}  \sum_{ u \in S(e) } \left( \sum_{v \in S(e) \setminus \{ u \} } cost(u,v) \right)^2
 - \frac{1}{2} \sum_{u \in S(e)} \sum_{v\in S(e) \setminus \{u\}} cost^2(u,v) \nonumber \\[0.2em]
&= \frac{1}{2} \ensuremath{\mathrm{QD}}(e)  - \frac{1}{2} \sum_{u \in S(e)} \sum_{v\in S(e) \setminus \{u\}} cost^2(u,v)
 \ . \label{eq::nasty_25}
\end{align}
The last sum in~(\ref{eq::nasty_25}) is equal to:
\begin{align}
& \sum_{u \in S(e)} \sum_{v\in S(e) \setminus \{u\}} cost^2(u,v) = \sum_{u \in S(e)}\ensuremath{\mathrm{SQ}}(u) - \ensuremath{\mathrm{SQ}}(e).
 \ . \label{eq::nasty_3}
\end{align} 
The value of the sum $\sum_{u \in S(e)}\ensuremath{\mathrm{SQ}}(u)$ can be computed for every edge $e$ in $\Theta(n)$ time 
in total as follows; for every tip $u \in S$ we store $\ensuremath{\mathrm{SQ}}(u)$ together with this tip, and then scan bottom-up
the tree adding those values that are in the subtree of each edge. 
For the remaining part of~(\ref{eq::nasty_2b}) we get:
\begin{align}
& \sum_{ \{u, v\} \in \Delta(S(e))} \hspace{0.05in} \sum_{ x \in S \setminus S(e) }  cost(u,x) \cdot cost(x,v)
\nonumber \\[0.2em] 
 &=\sum_{ \{u, v\} \in \Delta(S(e))} \hspace{0.05in} \sum_{ x \in S \setminus S(e) }
\left( cost(u,e) + cost(x,e) \right) \left( cost(v,e) + cost(x,e) \right) \nonumber \\[0.2em]
 &=\sum_{ \{u, v\} \in \Delta(S(e))} \hspace{0.05in} \sum_{ x \in S \setminus S(e) } cost(u,e) \cdot cost(v,e)
\nonumber \\[0.2em]
 &+\sum_{ \{u, v\} \in \Delta(S(e))} \hspace{0.05in} \sum_{ x \in S \setminus S(e) } cost(x,e) \cdot
(cost(u,e)+ cost(v,e)) \nonumber \\[0.2em]
 &+\sum_{ \{u, v\} \in \Delta(S(e))} \hspace{0.05in} \sum_{ x \in S \setminus S(e) } cost^2(x,e)
 \ . \label{eq::nasty_4}
\end{align}
The first sum in~(\ref{eq::nasty_4}) is equal to:
\begin{align}
\sum_{{\{u, v\}\in\Delta(S(e))}} \hspace{0.05in} \sum_{ x \in S \setminus S(e) } \!\!cost(u,e) \cdot cost(v,e) =
(s-s(e))\left( \ensuremath{\mathrm{TC_{sub}}}^2(e) - \ensuremath{\mathrm{SQ_{sub}}}(e) \right)\,. \label{eq::nasty_5}
\end{align}
For the second sum in~(\ref{eq::nasty_4}) we have:
{\small
\begin{align}
& \sum_{ \{u, v\} \in \Delta(S(e))} \hspace{0.05in} \sum_{ x \in S \setminus S(e) } cost(x,e) \cdot 
\big(cost(u,e)+ cost(v,e)\big) \nonumber \\[0.2em]
 = & \big(s(e)-1\big)\!\!\sum_{x \in S \setminus S(e) }\!\!cost(x,e) \cdot \ensuremath{\mathrm{TC_{sub}}}(e)
 = \big(s(e)-1\big) \cdot \ensuremath{\mathrm{TC_{sub}}}(e) \cdot \big(\ensuremath{\mathrm{PC}}(e) - \ensuremath{\mathrm{TC_{sub}}}(e) \big)\,. \label{eq::nasty_6}
\end{align}
}
The last sum in~(\ref{eq::nasty_4}) can be written as:
\begin{align}
& \sum_{ \{u, v\} \in \Delta(S(e))} \hspace{0.05in} \sum_{ x \in S \setminus S(e) } cost^2(x,e)
= \frac{s(e)(s(e)-1)}{2}\left( \ensuremath{\mathrm{PSQ}}(e) - \ensuremath{\mathrm{SQ_{sub}}}(e) \right)
\ . \label{eq::nasty_7}
\end{align}
Combining the analyses that we did from (\ref{eq::nasty_0}) up to (\ref{eq::nasty_7}) we get:
%
\begin{align*}
 \ensuremath{\mathrm{TRS}}(C) & = \frac{1}{6} \sum_{e \in E} w_e \sum_{u \in S(e)} \left( \ensuremath{\mathrm{SM}}(u) - \frac{3}{2}\ensuremath{\mathrm{SQ}}(u) \right) 
+ \frac{1}{2} \cdot \ensuremath{\mathrm{QD}}(e) + \frac{1}{2} \cdot \ensuremath{\mathrm{SQ}}(e) \nonumber \\[0.2em]
& + (s-2s(e)+1) \cdot \ensuremath{\mathrm{TC_{sub}}}^2(e) 
- \frac{2s - 2\cdot s(e)+ s(e)(s(e)-1)}{2} \cdot \ensuremath{\mathrm{SQ_{sub}}}(e) \nonumber \\[0.2em]
& + (s(e)-1) \cdot \ensuremath{\mathrm{TC_{sub}}}(e) \cdot \ensuremath{\mathrm{PC}}(e) + \frac{s(e)(s(e)-1)}{2} \cdot \ensuremath{\mathrm{PSQ}}(e)\; .
\end{align*}
%
%
The value of $\ensuremath{\mathrm{TRS}}(D)$ can be expressed as:
\begin{align*}
& \sum_{ u \in S } \sum_{ \substack{ v,x,y \in S \setminus \{u\} \\ v,x,y \textrm{ are distinct }  } }
cost(u,v) \cdot cost(u,x) \cdot cost(u,y)\\
&= \frac{1}{6} \bigg( \sum_{u \in S} \ensuremath{\mathrm{TC}}^3(u)
- 2 \cdot \ensuremath{\mathrm{TRS}}(A) - 3 \cdot \ensuremath{\mathrm{TRS}}(B) \bigg) \\[0.2em]
 &=\frac{1}{6} \cdot \sum_{u \in S} \ensuremath{\mathrm{TC}}^3(u) + \frac{2}{3} \cdot \ensuremath{\mathrm{CB}}(\mathcal{T}) - \frac{1}{2} \cdot \ensuremath{\mathrm{SQ}}(u) \cdot \ensuremath{\mathrm{TC}}(u)  \ .
\end{align*}
For $\ensuremath{\mathrm{TRS}}(E)$ we get:
\begin{align*}
&\sum_{ \{u,v\} \in \Delta(S) } \hspace{0.05in} \sum_{ \{x,y \} \in \Delta(S \setminus \{u,v\} ) } cost^2(u,v) \cdot cost(x,y)
\\[0.2em]
 &= \sum_{ \{u,v\} \in \Delta(S) } cost^2(u,v) ( \ensuremath{\mathrm{TC}}(\mathcal{T}) - \ensuremath{\mathrm{TC}}(u) - \ensuremath{\mathrm{TC}}(v) + cost(u,v) ) \\[0.2em]
 &=  \ensuremath{\mathrm{TC}}(\mathcal{T}) \sum_{ e \in E } w_e \cdot \ensuremath{\mathrm{TC}}(e)  - \sum_{u \in S} \left( \ensuremath{\mathrm{SQ}}(u) \cdot \ensuremath{\mathrm{TC}}(u) \right) + \ensuremath{\mathrm{CB}}(\mathcal{T}) \ .
\end{align*}
We can rewrite $\ensuremath{\mathrm{TRS}}(F)$ as follows:
%
\begin{align*}
&  \sum_{{\{u,v\}\in \Delta(S)}} \!\! cost(u,v)\Bigg( \ensuremath{\mathrm{TC}}(u) \cdot \ensuremath{\mathrm{TC}}(v) - cost^2(u,v)
-\!\!\!\!\sum_{x \in S \setminus \{u,v\} }\!\!\!\!cost(u,x)\cdot cost(x,v) \!\Bigg) \\[0.6em]
 &=\!\sum_{ \{u,v\} \in \Delta(S) } cost(u,v) \cdot \ensuremath{\mathrm{TC}}(u) \cdot \ensuremath{\mathrm{TC}}(v) - \ensuremath{\mathrm{CB}}(\mathcal{T}) - 3 \cdot \ensuremath{\mathrm{TRS}}(C) \\[0.6em]
 &=  \sum_{ e \in E } w_e \cdot \ensuremath{\mathrm{Mult}}(e) - \ensuremath{\mathrm{CB}}(\mathcal{T}) - 3 \cdot \ensuremath{\mathrm{TRS}}(C) \ .
\end{align*}
%
%
For the value of $\ensuremath{\mathrm{TRS}}(G)$ we have:
\begin{align}
\ensuremath{\mathrm{TRS}}(G) =&  \frac{1}{2}\sum_{{\{u,v\}\in \Delta(S)}} cost(u,v) \sum_{x \in S \setminus \{u,v\} }
\big( cost(u,x) + cost(v,x) \big) \Bigg(\ensuremath{\mathrm{TC}}(\mathcal{T}) \nonumber \\[0.2em]
&- \ensuremath{\mathrm{TC}}(u) -\ensuremath{\mathrm{TC}}(v) - \ensuremath{\mathrm{TC}}(x) + cost(u,v) + cost(u,x) + cost(v,x) \Bigg) \ . \label{eq::case_G}
\end{align}
We now break the sum in~(\ref{eq::case_G}) into five pieces and express each piece of 
this sum in terms of the quantities in Table~\ref{tab::quantities}. The first piece of 
the sum is equal to:
\begin{align*}
& \frac{1}{2} \sum_{ \{u,v\} \in \Delta(S) } cost(u,v) \sum_{x \in S \setminus \{u,v\} } 
\left( cost(u,x) + cost(v,x) \right) \cdot \ensuremath{\mathrm{TC}}(\mathcal{T}) \nonumber \\[0.8em]
  &=  \frac{1}{2} \cdot \ensuremath{\mathrm{TC}}(\mathcal{T})  \sum_{u \in S}\ensuremath{\mathrm{TC}}^2(u) - \sum_{ \{ u,v\} \in \Delta(S)} cost^2(u,v)\\[.6em]
   &= \frac{1}{2} \cdot \ensuremath{\mathrm{TC}}(\mathcal{T})  \sum_{u \in S}\ensuremath{\mathrm{TC}}^2(u) - \sum_{e \in E} w_e \cdot \ensuremath{\mathrm{TC}}(e) \ . 
\end{align*}
The second piece that we take from the sum in~(\ref{eq::case_G}) can be expressed as:
\begin{align}
& -\frac{1}{2} \sum_{ \{u,v\} \in \Delta(S) } cost(u,v) \sum_{x \in S \setminus \{u,v\} } 
\left( cost(u,x) + cost(v,x) \right)  \left( \ensuremath{\mathrm{TC}}(u) +\ensuremath{\mathrm{TC}}(v) \right) \nonumber \\[.2em]
=&-\frac{1}{2} \sum_{ \{u,v\} \in \Delta(S) } cost(u,v) \left( \ensuremath{\mathrm{TC}}(u) +\ensuremath{\mathrm{TC}}(v) - 2 \cdot cost(u,v) \right)  \left( \ensuremath{\mathrm{TC}}(u) +\ensuremath{\mathrm{TC}}(v) \right) \nonumber \\[.2em]
=&-\frac{1}{2} \sum_{ \{u,v\} \in \Delta(S) } cost(u,v) \Big( \ensuremath{\mathrm{TC}}^2(u) + \ensuremath{\mathrm{TC}}^2(u)
+ 2 \cdot \ensuremath{\mathrm{TC}}(u) \cdot \ensuremath{\mathrm{TC}}(v) \nonumber \\
&-2\cdot cost(u,v) \cdot \big(\ensuremath{\mathrm{TC}}(u) + \ensuremath{\mathrm{TC}}(v)\big) \Big)
\nonumber \\[.2em]
=&-\frac{1}{2} \sum_{u \in S }\ensuremath{\mathrm{TC}}^3(u) -  \sum_{ \{v,x\} \in \Delta(S) } cost(v,x) \cdot \ensuremath{\mathrm{TC}}(v) \cdot \ensuremath{\mathrm{TC}}(x) \nonumber \\
&+\sum_{ \{y,z\} \in \Delta(S) } cost^2(y,z) \big( \ensuremath{\mathrm{TC}}(y) + \ensuremath{\mathrm{TC}}(z) \big) \nonumber \\[.2em]
=&-\frac{1}{2} \sum_{u \in S }\ensuremath{\mathrm{TC}}^3(u) - \sum_{ e \in E } w_e \cdot \ensuremath{\mathrm{Mult}}(e)
+ \sum_{ u \in S } \ensuremath{\mathrm{SQ}}(u) \cdot \ensuremath{\mathrm{TC}}(u) \ . \label{eq::case_G_2}
\end{align}
The next piece that we select from~(\ref{eq::case_G}) is equal to:
\begin{align}
&-\frac{1}{2} \sum_{ \{u,v\} \in \Delta(S) } cost(u,v) \sum_{x \in S \setminus \{u,v\} }
\big( cost(u,x) + cost(v,x) \big) \cdot \ensuremath{\mathrm{TC}}(x) \nonumber \\
=&-\frac{1}{2} \sum_{ u \in S } \hspace{0.05in} \sum_{ v \in S \setminus \{u\} } cost(u,v)
\sum_{x \in S \setminus \{u\}} cost(u,x) \cdot \ensuremath{\mathrm{TC}}(x) \nonumber \\
&+\frac{1}{2} \sum_{u \in S} \hspace{0.05in}
\sum_{ v \in S \setminus \{u\} } cost^2(u,v) \cdot \ensuremath{\mathrm{TC}}(v)  \nonumber \\[.2em]
=&-\frac{1}{2} \sum_{ u \in S } \ensuremath{\mathrm{TC}}(u) \cdot \ensuremath{\mathrm{SM}}(u) +
\frac{1}{2} \sum_{ \{u,v\} \in \Delta(S) } cost^2(u,v) \left( \ensuremath{\mathrm{TC}}(u) + \ensuremath{\mathrm{TC}}(v) \right)
\nonumber \\[.2em]
=&\frac{1}{2} \sum_{ u \in S }  \ensuremath{\mathrm{TC}}(u) \big(\ensuremath{\mathrm{SQ}}(u) -  \ensuremath{\mathrm{SM}}(u) \big)  \ .  \label{eq::case_G_3}
\end{align}
For the fourth piece of the sum in~(\ref{eq::case_G}) we get:
\begin{align}
& \frac{1}{2} \sum_{ \{u,v\} \in \Delta(S) } cost^2(u,v) 
\sum_{x \in S \setminus \{u,v\} } cost(u,x) + cost(v,x) \\[0.1in]
= & \frac{1}{2} \cdot \ensuremath{\mathrm{TRS}}(B) 
 =  \frac{1}{2} \sum_{ u \in S } \ensuremath{\mathrm{SQ}}(u) \cdot \ensuremath{\mathrm{TC}}(u) - \ensuremath{\mathrm{CB}}(\mathcal{T}) \ . \label{eq::case_G_4}
\end{align}
The last piece of the sum in~(\ref{eq::case_G}) can be expressed as:
\begin{align}
& \frac{1}{2} \sum_{ \{u,v\} \in \Delta(S) } cost(u,v) \sum_{x \in S \setminus \{u,v\} } 
\left( cost(u,x) + cost(v,x) \right)^2
\nonumber \\[0.1in]
 = & \frac{1}{2} \sum_{ \{u,v\} \in \Delta(S) } cost(u,v) \sum_{x \in S \setminus \{u,v\} } 
\left( cost^2(u,x) + cost^2(v,x) \right)
+ 3\cdot \ensuremath{\mathrm{TRS}}(C) \nonumber \\[0.1in]
 = & \frac{1}{2} \sum_{ \{u,v\} \in \Delta(S) } cost(u,v) 
\left( \ensuremath{\mathrm{SQ}}(u) + \ensuremath{\mathrm{SQ}}(v) - 2 \cdot cost^2(u,v) \right) + 3 \cdot \ensuremath{\mathrm{TRS}}(C) 
\nonumber \\[0.1in]
 = & \frac{1}{2} \sum_{u \in S} \ensuremath{\mathrm{SQ}}(u) \cdot \ensuremath{\mathrm{TC}}(u) -  \ensuremath{\mathrm{CB}}(\mathcal{T})  + 3 \cdot \ensuremath{\mathrm{TRS}}(C) \ .
\label{eq::case_G_5}
\end{align}
Combining our analyses from (\ref{eq::case_G}) up to (\ref{eq::case_G_5}) we get:
%
\begin{align*}
 \ensuremath{\mathrm{TRS}}(G) & =  \frac{1}{2} \cdot \ensuremath{\mathrm{TC}}(\mathcal{T}) \sum_{u \in S}\ensuremath{\mathrm{TC}}^2(u) -\sum_{e \in E} w_e \left( \ensuremath{\mathrm{TC}}(e) + \ensuremath{\mathrm{Mult}}(e) \right)
\nonumber \\[0.1in]
& + \frac{1}{2} \sum_{u \in S}\ensuremath{\mathrm{TC}}(u)\!\cdot\!\big( 5\cdot \ensuremath{\mathrm{SQ}}(u)\!-\!\ensuremath{\mathrm{SM}}(u)-\!\ensuremath{\mathrm{TC}}^2(u) \big)
-2 \cdot\!\ensuremath{\mathrm{CB}}(T)+\!3\!\cdot\!\ensuremath{\mathrm{TRS}}(C)\,.
\end{align*}
%
%
We can express $\ensuremath{\mathrm{TRS}}(H)$ using the values of the other isomorphism classes:
\begin{align*}
& \ensuremath{\mathrm{TRS}}(H) = \frac{1}{6} \sum_{{\{u,v\}\in\Delta(S)}} \sum_{{\{x,y\}\in \Delta(S)}}
\sum_{{\{c,d\}\in\Delta(S)}} cost(u,v) \cdot cost(x,y) \cdot cost(c,d) \\[0.1in]
&- \ensuremath{\mathrm{TRS}}(A) - 3 \cdot \ensuremath{\mathrm{TRS}}(B) - 6 \cdot\!\ensuremath{\mathrm{TRS}}(C) - 6 \cdot \ensuremath{\mathrm{TRS}}(D) \\[0.1in] 
& - 3 \cdot \ensuremath{\mathrm{TRS}}(E) -6 \cdot \ensuremath{\mathrm{TRS}}(F) 
- 6 \cdot \ensuremath{\mathrm{TRS}}(G) \nonumber \\[0.1in]
& = \frac{1}{6} \cdot \ensuremath{\mathrm{TC}}^3(\mathcal{T}) -  \frac{1}{6} \cdot \ensuremath{\mathrm{TRS}}(A) - \frac{1}{2} \cdot \ensuremath{\mathrm{TRS}}(B) 
- \ensuremath{\mathrm{TRS}}(C) - \ensuremath{\mathrm{TRS}}(D)  \\[0.1in]
& -\frac{1}{2} \cdot \ensuremath{\mathrm{TRS}}(E)
 - \ensuremath{\mathrm{TRS}}(F) - \ensuremath{\mathrm{TRS}}(G) \ .
\end{align*}

We get the value of $E_{R \in \mathrm{Sub}(S, r)}[\ensuremath{\mathrm{MPD}}^3(\mathcal{T},R)]$ by plugging into (\ref{eq::le_big_sum}) 
the values that we got for all eight isomorphism classes of triples. For any isomorphism class $X$ 
we showed that the value $\ensuremath{\mathrm{TRS}}(X)$ can be computed by using the quantities in Table~\ref{tab::quantities}. 
The lemma follows from the fact that each quantity that appears in this table is used a constant number
of times for computing value $\ensuremath{\mathrm{TRS}}(X)$ for any class $X$, and since we showed that we can precompute all these 
quantities in $\Theta(n)$ time in total. \qed
\end{proof}

\begin{theorem}
Let $\mathcal{T}$ be a phylogenetic tree that contains $s$ tips, and let $r$ be a natural number with $r \leq s$.
The skewness of the mean pairwise distance on $\mathcal{T}$ among all subsets of exactly $r$ tips of $\mathcal{T}$
can be computed in  $\Theta(n)$ time.
\end{theorem}
\begin{proof}
According to the definition of skewness, as it is also presented in~(\ref{eq::skewness_basic}),
we need to prove that we can compute in $\Theta(n)$ time the expectation and the variance of 
the $\ensuremath{\mathrm{MPD}}$, and the value of the expression $E_{R \in \mathrm{Sub}(S, r)}[\ensuremath{\mathrm{MPD}}^3(\mathcal{T},R)]$.
In a previous paper we showed that the expectation and the variance of the $\ensuremath{\mathrm{MPD}}$ can be computed
in $\Theta(n)$ time. By combining this with
Lemma~\ref{le::main_lemma} we get the proof of the theorem. \qed
\end{proof}



\end{document}